\documentclass[conference]{IEEEtran}
\IEEEoverridecommandlockouts

\usepackage{cite}
\usepackage{amsmath,amssymb,amsfonts}
\usepackage{graphicx}
\usepackage{textcomp}
\usepackage{xcolor}

\usepackage{amsthm}
\usepackage{float,epsfig}
\usepackage{rotating}
\usepackage{subfigure}
\usepackage{subcaption} 
\usepackage{caption}
\usepackage{tkz-graph}
\usepackage{color}
\usepackage{algorithm}
\usepackage{multicol}
\usepackage[utf8]{inputenc}
\usepackage{braket}
\usepackage{multirow}
\usepackage{qcircuit}
\usepackage[overload]{empheq}
\usepackage{hyperref}
\usepackage{color,lscape,enumerate}

\usepackage[noend]{algpseudocode}
\usepackage{mathtools}
\DeclarePairedDelimiter\ceil{\lceil}{\rceil}

\usepackage{hyperref} 
\usepackage{xparse}
\usepackage{lipsum}



\begin{document}

\newtheorem{theorem}{\bf Theorem}[section]
\newtheorem{proposition}[theorem]{\bf Proposition}
\newtheorem{definition}[theorem]{\bf Definition}
\newtheorem{corollary}[theorem]{\bf Corollary}
\newtheorem{example}[theorem]{\bf Example}
\newtheorem{exam}[theorem]{\bf Example}
\newtheorem{remark}[theorem]{\bf Remark}
\newtheorem{lemma}[theorem]{\bf Lemma}
\newcommand{\nrm}[1]{|\!|\!| {#1} |\!|\!|}

\newcommand{\calL}{{\mathcal L}}
\newcommand{\calX}{{\mathcal X}}
\newcommand{\calY}{{\mathcal Y}}
\newcommand{\calZ}{{\mathcal Z}}
\newcommand{\calW}{{\mathcal W}}
\newcommand{\calA}{{\mathcal A}}
\newcommand{\calB}{{\mathcal B}}
\newcommand{\calC}{{\mathcal C}}
\newcommand{\calK}{{\mathcal K}}
\newcommand{\C}{{\mathbb C}}
\newcommand{\Z}{{\mathbb Z}}
\newcommand{\R}{{\mathbb R}}
\renewcommand{\SS}{{\mathbb S}}
\newcommand{\LL}{{\mathbb L}}
\newcommand{\st}{{\star}}
\def\kernel{\mathop{\rm kernel}\nolimits}
\def\sigan{\mathop{\rm span}\nolimits}

\newcommand{\klasse}{{\boldsymbol \Delta}}

\newcommand{\ba}{\begin{array}}
\newcommand{\ea}{\end{array}}
\newcommand{\von}{\vskip 1ex}
\newcommand{\vone}{\vskip 2ex}
\newcommand{\vtwo}{\vskip 4ex}
\newcommand{\dm}[1]{ {\displaystyle{#1} } }

\newcommand{\be}{\begin{equation}}
\newcommand{\ee}{\end{equation}}
\newcommand{\beano}{\begin{eqnarray*}}
\newcommand{\eeano}{\end{eqnarray*}}
\newcommand{\inp}[2]{\langle {#1} ,\,{#2} \rangle}
\def\bmatrix#1{\left[ \begin{matrix} #1 \end{matrix} \right]}
\def\cmatrix#1{\left( \begin{matrix} #1 \end{matrix} \right)}
\def \noin{\noindent}
\newcommand{\evenindex}{\Pi_e}



\def \R{{\mathbb R}}
\def \C{{\mathbb C}}
\def \F{{\mathbb F}}
\def \K{{\mathbb K}}
\def \H{{\mathbb H}}
\def \cu{\mathrm{CU}}

\def \T{{\mathbb T}}
\def \Pb{\mathrm{P}}
\def \N{{\mathbb N}}
\def \Ib{\mathrm{I}}
\def \Ls{{\Lambda}_{m-1}}
\def \Gb{\mathrm{G}}
\def \Hb{\mathrm{H}}
\def \Lam{{\Lambda}}

\def \Qb{\mathrm{Q}}
\def \Rb{\mathrm{R}}
\def \Mb{\mathrm{M}}
\def \norm{\nrm{\cdot}\equiv \nrm{\cdot}}

\def \A{{{\mathbb P}_1(\C^{n\times n})}}
\def \H{{\mathbb H}}
\def \L{{\mathcal L}}
\def \G{{\F_{\tt{H}}}}
\def \S{\mathbb{S}}
\def \s{\mathbb{s}}
\def \sigmin{\sigma_{\min}}
\def \elam{\Lambda_{\epsilon}}
\def \slam{\Lambda^{\S}_{\epsilon}}
\def \Ib{\mathrm{I}}
\def \Tb{\mathrm{T}}
\def \d{{\delta}}

\def \Lb{\mathrm{L}}
\def \N{{\mathbb N}}
\def \Ls{{\Lambda}_{m-1}}
\def \Gb{\mathrm{G}}
\def \Hb{\mathrm{H}}
\def \Delta{\triangle}
\def \Rar{\Rightarrow}
\def \p{{\mathsf{p}(\lam; v)}}

\def \D{{\mathbb D}}

\def \tr{\mathrm{Tr}}
\def \cond{\mathrm{cond}}
\def \lam{\lambda}
\def \sig{\sigma}
\def \sign{\mathrm{sign}}

\def \ep{\epsilon}
\def \diag{\mathrm{diag}}
\def \rev{\mathrm{rev}}
\def \vec{\mathrm{vec}}

\def \ham{\mathsf{Ham}}
\def \herm{\mathsf{Herm}}
\def \sym{\mathsf{sym}}
\def \odd{\mathsf{sym}}
\def \en{\mathrm{even}}
\def \rank{\mathrm{rank}}
\def \pf{{\bf Proof: }}
\def \dist{\mathrm{dist}}
\def \rar{\rightarrow}

\def \rank{\mathrm{rank}}
\def \pf{{\bf Proof: }}
\def \dist{\mathrm{dist}}
\def \Re{\mathsf{Re}}
\def \Im{\mathsf{Im}}
\def \re{\mathsf{re}}
\def \im{\mathsf{im}}

\def \sym{\mathsf{CSym}}
\def \sksym{\mathsf{skew\mbox{-}sym}}
\def \odd{\mathrm{odd}}
\def \even{\mathrm{even}}
\def \herm{\mathsf{Herm}}
\def \skherm{\mathsf{skew\mbox{-}Herm}}
\def \str{\mathrm{ Struct}}
\def \cnot{\mathrm{CNOT}}
\def \eproof{$\blacksquare$}

\def \bS{{\bf S}}
\def \cA{{\cal A}}
\def \E{{\mathcal E}}
\def \X{{\mathcal X}}
\def \cH{\mathcal{H}}
\def \cJ{\mathcal{J}}
\def \tr{\mathrm{Tr}}
\def \range{\mathrm{Range}}
\def \adj{\star}

\def \pal{\mathrm{palindromic}}
\def \palpen{\mathrm{palindromic~~ pencil}}
\def \palpoly{\mathrm{palindromic~~ polynomial}}
\def \odd{\mathrm{odd}}
\def \even{\mathrm{even}}
\def \QT{{\texttt{QT}}}

\newcommand{\tm}[1]{\textcolor{magenta}{ #1}}
\newcommand{\tre}[1]{\textcolor{red}{ #1}}
\newcommand{\tb}[1]{\textcolor{blue}{ #1}}
\newcommand{\tg}[1]{\textcolor{green}{ #1}}

\def\BibTeX{{\rm B\kern-.05em{\sc i\kern-.025em b}\kern-.08em
    T\kern-.1667em\lower.7ex\hbox{E}\kern-.125emX}}

\title{Quantum Learning of Classical Correlations with \\ continuous-domain Pauli Correlation Encoding}

\author{\IEEEauthorblockN{ Vicente P. Soloviev}
\IEEEauthorblockA{\textit{Fujitsu Research of Europe Ltd.} \\
Madrid, Spain \\
vicente.perezsoloviev@fujitsu.com}
\and
\IEEEauthorblockN{ Bibhas Adhikari}
\IEEEauthorblockA{\textit{Fujitsu Research of America Inc.} \\
Santa Clara, CA, USA \\
badhikari@fujitsu.com}
}

\maketitle

\begin{abstract}
We propose a quantum machine learning framework for estimating classical covariance matrices using parameterized quantum circuits within the Pauli-Correlation-Encoding (PCE) paradigm. We introduce two quantum covariance estimators: the \emph{C-Estimator}, which constructs the covariance matrix through a Cholesky factorization to enforce positive (semi)definiteness, and a computationally efficient \emph{E-Estimator}, which directly estimates covariance entries from observable expectation values. We analyze the trade-offs between the two estimators in terms of qubit requirements and learning complexity, and derive sufficient conditions on regularization parameters to ensure positive (semi)definiteness of the estimators. Furthermore, we show that the barren plateau phenomenon in training the variational quantum circuit for E-estimator can be mitigated by appropriately choosing the regularization parameters in the loss function for HEA ansatz. The proposed framework is evaluated through numerical simulations using randomly generated covariance matrices. We examine the convergence behavior of the estimators, their sensitivity to low-rank assumptions, and their performance in covariance completion with partially observed matrices. The results indicate that the proposed estimators provide a robust approach for learning covariance matrices and offer a promising direction for applying quantum machine learning techniques to high-dimensional statistical estimation problems.
\end{abstract}


\section{Introduction}
Estimation of covariance and correlation matrices is a fundamental task in multivariate statistics and machine learning with wide applications in finance, signal processing, and network analysis \cite{jolliffe2002principal,pourahmadi2013high,tsukuma2020shrinkage,masuda2025introduction,garcia2025high,adhikari2025signed}. The classical estimator is the sample covariance matrix, which is consistent when the number of observations is large compared to the number of variables. However, in many modern applications such as financial markets, the number of assets can be comparable to or larger than the number of observations, leading to unstable and poorly conditioned estimates. To address this issue, several regularized estimators have been proposed, including shrinkage estimators that combine the sample covariance matrix with a structured target \cite{ledoit2004well}, factor-based covariance models \cite{fan2008high}, and sparse covariance estimators based on thresholding techniques \cite{bickel2008covariance}. These approaches aim to improve estimation accuracy and numerical stability, particularly in high-dimensional settings where the sample covariance matrix becomes singular or highly noisy.

Despite their usefulness, several challenges remain in estimating covariance and correlation matrices in large-scale problems. One common approach is to assume a low-rank structure motivated by latent factor models, where the covariance matrix is approximated by a low-rank component plus a diagonal matrix capturing idiosyncratic noise \cite{fan2013large}. While low-rank models significantly reduce the number of parameters, they may fail to capture complex dependence structures present in real data. However, low-rank covariance models are widely used in applications such as financial econometrics, where asset returns are often modeled through a small set of market factors, enabling more robust estimation of risk and portfolio allocation \cite{fan2008high}. Another major challenge arises in estimating the inverse covariance matrix (precision matrix), which plays a key role in graphical models and portfolio analysis \cite{stevens1998inverse}. Sparse precision matrix estimation methods such as the graphical lasso have been proposed to address this issue \cite{friedman2008sparse,yuan2010high}. Furthermore, in many practical situations covariance matrices are only partially observed due to missing data or limited observations, leading to the problem of covariance matrix completion \cite{lounici2014high}. In financial time-series data, for example, missing values may occur because different assets are traded at different times or due to gaps in historical records. When the underlying covariance matrix cannot be estimated directly from observations, specialized techniques are used to recover or approximate the missing entries. Techniques based on low-rank matrix completion and convex optimization have been developed to recover missing entries under structural assumptions \cite{candes2009exact}. These methods aim to reconstruct a positive semidefinite covariance estimator while preserving the statistical dependence structure among the variables. The afore mentioned challenges highlight the need for robust estimation methods capable of handling high-dimensional, noisy, and incomplete data.

In this paper, we introduce a quantum machine learning model to estimate the pairwise covariance or correlation strengths of a collection of random variables. For a set of $n$ random variables $\{X_1,X_2,\ldots,X_n\}$, the covariance matrix $\Sigma=[\Sigma_{ij}]$ is symmetric and requires the estimation of ${n\choose 2}$ parameters $\Sigma_{ij}$ for $i>j$. In our proposal, we employ the recently introduced Pauli-Correlation-Encoding (PCE) framework for quantum learning of a classical covariance estimator. The PCE technique was originally proposed to address binary optimization problems such as the Max-Cut problem on combinatorial graphs \cite{sciorilli2025towards} (see Section~\ref{Sec:PCE}). In this work, we extend the PCE framework to a continuous setting, where the variables of the optimization problem are real-valued parameters encoded through Pauli-string observables. Consequently, the variables correspond to the expectation values of these observables that optimize a loss function defined as the Frobenius distance between the proposed quantum circuit-based estimator and a given classical estimator.

More specifically, we consider a parameterized quantum circuit $\mathsf{U}(\boldsymbol{\theta})$ with parameters $\boldsymbol{\theta}$ acting on a $\eta$-qubit quantum register. The parameters $\Sigma_{ij}$ are encoded through the expectation values of ${n\choose 2}$ observables evaluated on the output quantum state 
$\mathsf{U}(\boldsymbol{\theta})\ket{0}^{\otimes \eta}=\ket{\psi(\boldsymbol{\theta})}.$ The learning of the parameters is achieved by minimizing a loss function $\mathcal{L}(\boldsymbol{\theta})$ defined through the Frobenius distance between the classical covariance estimator and the quantum estimator derived from the expectation values of the observables. The overall procedure of the proposed method is illustrated in Figure~\ref{fig:scheme}.

\begin{figure}[htbp]
    \centering
    \includegraphics[width=1.0\linewidth]{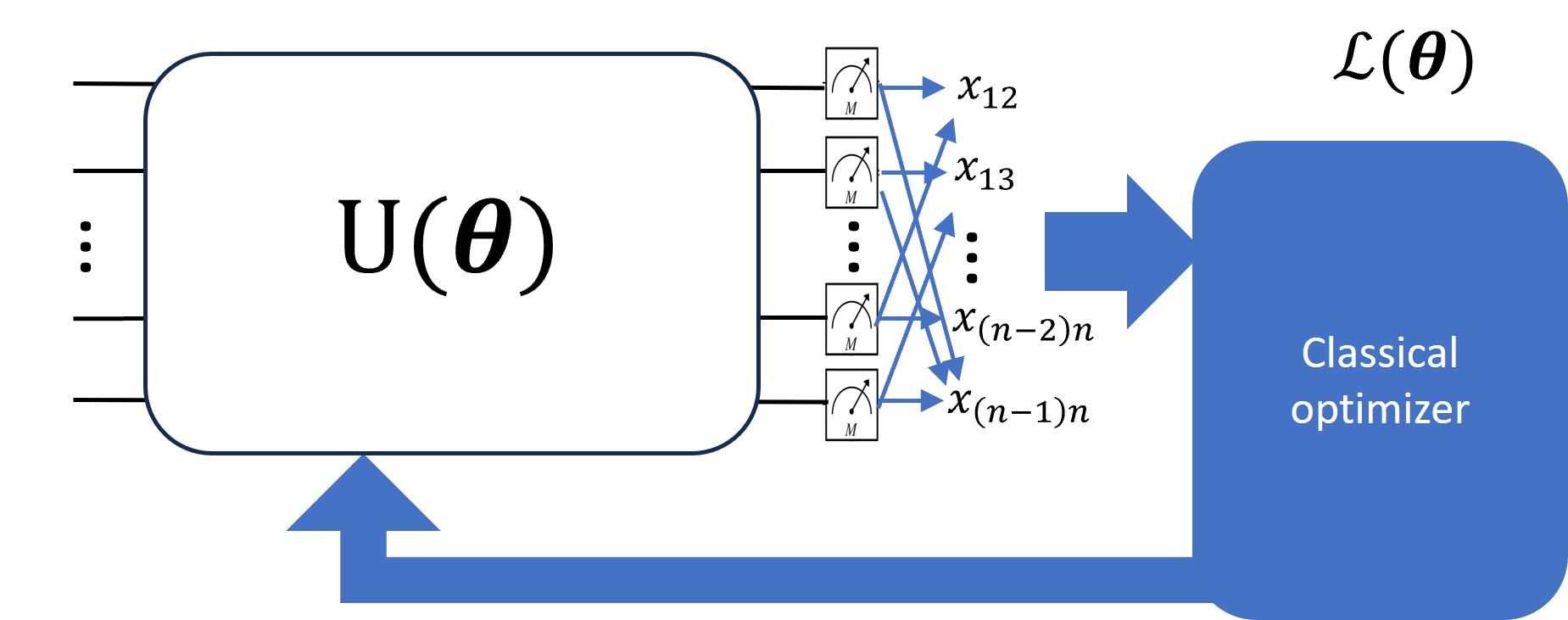}
    \caption{Proposed quantum-classical workflow}
    \label{fig:scheme}
\end{figure}

We address three fundamental problems in covariance estimation: (a) low-rank approximation of a covariance matrix, (b) completion of a covariance matrix, and (c) estimation of the inverse (precision) matrix. To this end, we propose two covariance estimators based on a parameterized quantum circuit (PQC): (i) the \emph{C-Estimator}, and (ii) the \emph{E-Estimator}.

The C-Estimator is defined through the Cholesky factorization of the proposed estimator of the form $\widehat{\Sigma}=LL^{\top}$. The entries of the upper triangular matrix $L^{\top}$ are defined as $l_{ij}=c_{ij}x_{ij}$ for $i\geq j$, where $x_{ij}=\langle \psi(\boldsymbol{\theta})|\Pi_{ij}|\psi(\boldsymbol{\theta})\rangle$. The diagonal entries $l_{ii}$ are chosen as functions of $\mathrm{Var}(X_i)$ together with $x_{il}$ and $c_{li}$ for $l<i$, in such a way that the $i$-th diagonal entry of $LL^{\top}$ equals $\mathrm{Var}(X_i)$, where $X_i$ is the $i$-th random variable of the random vector $X=(X_1,\ldots,X_n)$. Here $c_{ij}$ denote regularization parameters. The number of qubits required in this model is $\eta=\left\lceil\sqrt{{n\choose 2}}\right\rceil$, and the observables $\Pi_{ij}$ are chosen from the set $\{\Pi^{(2)}_1,\ldots,\Pi^{(2)}_{{\eta\choose 2}}\}$, where each $\Pi_l^{(2)}$ is a permutation of Pauli strings of the form $X^{\otimes 2}\otimes I^{\otimes(\eta-2)}$, $Y^{\otimes 2}\otimes I^{\otimes(\eta-2)}$, or $Z^{\otimes 2}\otimes I^{\otimes(\eta-2)}$, with $X,Y,Z,$ and $I$ denoting the Pauli matrices.

Next, we introduce the E-Estimator, which is defined on an $\eta=n$-qubit register with observables $\Pi_{ij}$ chosen as permutations of the Pauli string $X^{\otimes 2}\otimes I^{\otimes(n-2)}$. In this case the proposed estimator is the symmetric matrix $\widehat{\Sigma}=[\widehat{\Sigma}_{ij}]$ defined by $\widehat{\Sigma}_{ij}=c_{ij}x_{ij}$ for $i\neq j$, where $x_{ij}=\langle \psi(\boldsymbol{\theta})|\Pi_{ij}|\psi(\boldsymbol{\theta})\rangle$, and $\widehat{\Sigma}_{ii}=\mathrm{Var}(X_i)$. The parameters $c_{ij}$ again serve as regularization coefficients.

In general, the E-estimator is not necessarily positive (semi)definite for arbitrary choices of the regularization parameters. Therefore, we derive sufficient conditions on the parameters $c_{ij}$ under which the covariance estimator $\widehat{\Sigma}$ becomes positive definite or positive semidefinite. Moreover, there exists a trade-off between the two learning methods. The C-Estimator requires significantly fewer qubits, whereas the learning procedure for the E-Estimator is computationally simpler and efficient, since its entries correspond directly to scaled expectation values of the observables and the observables form one set of mutually commuting operators. In addition, estimation of the precision matrix from the C-Estimator can be performed efficiently using the Cholesky factor $L$, thereby bypassing explicit matrix inversion. Furthermore, low-rank approximations can be obtained more efficiently within the C-Estimator framework by setting $n-r$ diagonal entries of $L$ to zero to obtain a rank-$r$ estimator, rather than computing a singular value decomposition or eigenvalue decomposition of the E-Estimator. However, the covariance completion problem can be efficiently addressed using both the estimators.

Finally, we analyze the E-Estimator algorithm to investigate the variance of the gradient of the loss function with respect to the quantum circuit parameters. In particular, we derive an explicit upper bound for the variance 
$\mathrm{Var}_{\boldsymbol{\theta}}\!\left(\partial \mathcal{L}/\partial \theta_{ij}\right)$ for HEA circuits with $poly(n)$ layers that forms an approximate $2$-design.
We show that this bound depends on the norm of the regularization vector formed by the parameters $c_{ij}$ as well as the Frobenius norm of the given classical covariance matrix. Our analysis indicates that the barren plateau phenomenon can be mitigated when some of the regularization parameters scale exponentially with $n$.

To evaluate the performance of the proposed quantum learning framework for classical covariance estimators and its application to covariance completion, we conduct the following experiments. We generate classical covariance estimators as random positive (semi)definite matrices of the form $AA^\top$, where the entries of $A$ follow a standard normal distribution. The convergence of the proposed learning framework is then examined using both the C-Estimator and the E-Estimator. The underlying PCE circuit is implemented using a hardware-efficient ansatz (HEA) for multi-qubit systems with varying numbers of qubits. 

Furthermore, we investigate low-rank covariance estimation using the C-Estimator for different target ranks $r\geq 1$ for randomly generated covariance matrices of order $n$, thereby analyzing the sensitivity of the method under low-rank assumptions. Finally, we evaluate the effectiveness of the proposed estimators for covariance completion by masking $O(n)$ entries of randomly generated covariance matrices $AA^\top\in\mathbb{R}^{n\times n}$ and reconstructing the missing entries using the C-Estimator and E-Estimator. Our numerical experiments demonstrate that both estimators based on the cPCE framework are robust for learning positive (semi)definite covariance matrices.

The remainder of the paper is organized as follows. Section~\ref{Sec:prelim} briefly reviews the notions of PCE, HEA, and classical covariance estimators. Section~\ref{Sec:cPCEQL} introduces the cPCE framework and presents the C-Estimator and E-Estimator for quantum learning of classical covariance estimators, along with their applications to covariance completion and precision matrix estimation. Finally, Section~\ref{sec_results} reports the results of numerical simulations evaluating the performance of the proposed estimators.

\section{Preliminaries}\label{Sec:prelim}

\subsection{Hardware Efficient Ansatz (HEA)}

A parametrized quantum circuit (PQC) represents a sequence of parametrized quantum gates for implementation in a quantum register.  In general, a PQC is represented as 
\begin{equation}\label{eqn:uCircuit}
    \mathsf{U}(\boldsymbol{\theta}) =\prod_{l=1}^L \mathsf{U}(\boldsymbol{\theta}_l), \,\,  \mathsf{U}(\boldsymbol{\theta}_l)= \prod_{k=0}^K \exp(-iH_k\theta_{lk}),
\end{equation} where the index $l$ indicates a layer, $\boldsymbol{\theta}_l=\{\theta_{1k},\hdots,\theta_{lK}\}$ is the set of parameters in $l$-th layer. The set of Hermitian traceless matrices $H_k$, called the generators of the PQC. Now we recall from  that a cost function $\mathcal{L}_{\boldsymbol{\theta}}(\rho,O)$ is said to have a barren plateau when training the parameters in $\boldsymbol{\theta}$ if 
$$\mbox{Var}_{\boldsymbol{\theta}}\left[\partial \mathcal{L}(\boldsymbol{\theta})/\partial\theta_\mu\right]\leq F(n)$$ with
$F(n)\in \mathcal{O}(1/b^n)$ for some $b>1$ and the variance is defined with respect to the set of parameters $\boldsymbol{\theta}.$


The construction of a layer of the Hardware Efficient Ansatz (HEA) in a quantum register begins by applying single-qubit gates to all qubits  and subsequently introducing a sequence of entangling two-qubit gates \cite{kandala2017hardware}. The two-qubit gates need not be continuously parameterized, although they may be as considered in \cite{campos2021abrupt}. Repeating this layer several rounds yields a parametrized quantum circuit whose entangling capacity increases with the number of layers. In most scenarios, $CNOT$ or control-$Z$ gates are employed, arranged in a circular topology acting on qubit pairs 
$(i,i+1)$ modulo $n,$ the total number of qubits in the register \cite{mcclean2018barren}, which we consider in this paper as a test case as described in Figure \ref{fig_ansatz} for a $4$-qubit register. However, the entangling gates may also be applied in all-to-all qubits \cite{skolik2021layerwise}.

\begin{figure*}[t]
\centering
\begin{subfigure}
\centering
\includegraphics[width=0.25\textwidth]{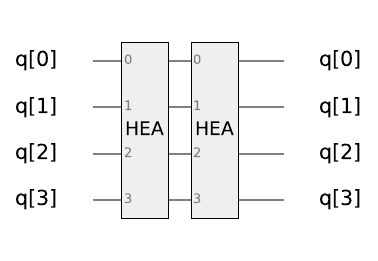}
\label{fig_ansatz_general}
\end{subfigure}
\begin{subfigure}
\centering
\includegraphics[width=0.72\textwidth]{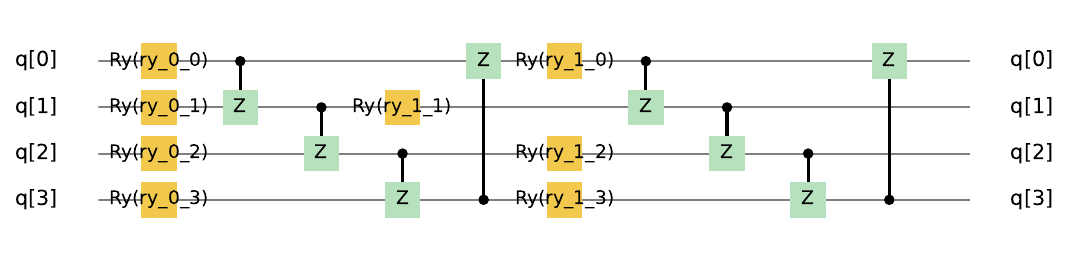}
\label{fig_ansatz_decomposed}
\end{subfigure}
\caption{Block-wise and gate-wise HEA built with a linear entangling structure and CZ gates, used in PCE strategy approach where $m = 10$, $n=4$, and $k=2$.}
\label{fig_ansatz}
\end{figure*}

\subsection{Pauli-Correlation-Encoding (PCE)}
\label{Sec:PCE}


First we briefly review the Pauli-correlation encoding (PCE) framework introduced in \cite{sciorilli2025towards} for combinatorial optimization problems with $m=O(n^k)$ binary variables using only $n$ qubits, for an  integer $k.$  Given the binary variables $x_j$ and traceless $n$-qubit Pauli strings $\Pi_j$ $1\leq j\leq m$, the PCE is defined as 
\begin{equation}\label{eqn:sgn}
    x_j:= \mbox{sgn}(\langle\Pi_j\rangle),
\end{equation} where $\langle\Pi_j\rangle :=\braket{\Psi|\Pi_j|\Psi}$ is the expectation value of $\Pi_j$ over a quantum state $\ket{\Psi}.$ In particular, to empower the expectation values computation, the authors suggest the encoding by considering Pauli-strings given by $\Pi^{(k)}=\left\{ \Pi_1^{(k)}, \Pi_2^{(k)}, \hdots, \Pi_m^{(k)}\right\}$ where each $\Pi_l^{(k)}$ is a permutation of either $X^{\otimes k}\otimes I^{\otimes n-k}$, $Y^{\otimes k}\otimes I^{\otimes n-k}$, or $Z^{\otimes k}\otimes I^{\otimes n-k}.$ The motivation behind this framework is to encode the binary variables through the sgn function as described in (\ref{eqn:sgn}). The proposed framework is applied to the weighted Max-Cut problem for a graph, which maximizes the objective function $\mathcal{V}(\boldsymbol{x})=\sum_{(i,j)\in E} W_{ij}(1-x_ix_j),$ where $\boldsymbol{x}$ denotes the set of binary variables $x_i,$ and $W_{ij}$ is the weigh of an edge $(i,j),$ $E$ denotes the edge set of the associated graph.

In the PCE setting, the state $\ket{\Psi}$ is defined as an output state of a parametrized circuit $U(\boldsymbol{\theta})$ with the set of parameters denoted as $\boldsymbol{\theta}$ and hence modeled as $\ket{\Psi(\boldsymbol{\theta})}=U(\boldsymbol{\theta})\ket{0}^{\otimes n}.$ In their variational approach, they consider the parametrized quantum circuit having a brickwise architecture with the loss function \begin{equation}
    \mathcal{L}=\sum_{(i,j)\in E} W_{ij} \tanh{(\alpha \langle\Pi_j\rangle) \tanh{(\alpha\langle\Pi_j\rangle)}} + \mathcal{L}^{(\mbox{reg})},
\end{equation} which is minimized to learn the parameters in $\boldsymbol{\theta}$. Here, the regularization term $$\mathcal{L}^{(\mbox{reg})}=\beta\nu \left[\frac{1}{m}\sum_{i\in V} \tanh (\alpha\langle\Pi_i\rangle)^2\right]^2$$
forces all correlations to have small values that improves the solver's performance, $\alpha >1$ is the rescaling factor which is shown to give a good performance for $\alpha\approx n^{[k/2]}$, $\nu=w(G)/2+w(T_{\min})/4$ with $w(G)$ and $w(T_{\min})$ the weights of the graph $G$ and its minimum spanning tree, respectively for weighted MaxCut and $\nu=|E|/2 + (m-1)/4$ for MaxCut. Once the output state is measured and a bit-string $\boldsymbol{x}$ is obtained using equation (\ref{eqn:sgn}), the final value of $\mathcal{V}(x^*)$ a classical post-processing step involving a single-bit swap search is performed for finding potential better solutions $x^*$, whose complexity is $O(|E|).$


Further, in their exploration of barren plateau characterization for the classical optimization, they show that random parameter initializations make the circuits sufficiently random, and the variance of $\mathcal{L}$ is given by $$\mbox{var}(\mathcal{L})=\frac{\alpha^4}{4^n}\sum_{(i,j)\in E} W_{ij}^2 +O\left(\frac{\alpha^6}{8^n}\right),$$ which is also empirically observed at circuit depths of about $8.5\times n.$ 

The PCE framework has been further applied to a variety of combinatorial optimization problems, including the Low Autocorrelation Binary Sequences (LABS) problem \cite{sciorilli2025competitive}, portfolio optimization \cite{soloviev2025large}, budget-constrained optimization \cite{martinez2026pauli}, and the Traveling Salesman Problem (TSP) \cite{do2026warm}.

\subsection{Classical covariance estimators}

The sample covariance estimator defined for a collection of random variables $X=(X_1,\ldots,X_n)$ with observations $\{x^{(t)}\}_{t=1}^{T}$ is given by
$$\widehat{\Sigma}=\frac{1}{T-1}\sum_{t=1}^{T}(x^{(t)}-\bar{x})(x^{(t)}-\bar{x})^{\top},$$
where $\bar{x}$ denotes the sample mean. The corresponding sample correlation estimator is obtained by normalizing the entries of $\widehat{\Sigma}$ by the estimated standard deviations. In high-dimensional problems the sample covariance matrix becomes noisy, ill-conditioned, or singular, motivating the development of improved estimators \cite{anderson2003introduction,fan2008high}.


Several regularization techniques have been proposed to address these issues. Shrinkage covariance estimators combine the sample covariance matrix with a structured target, such as the identity matrix or a factor-based model, to improve conditioning and estimation accuracy \cite{ledoit2004well,ledoit2020analytical}. Regularized covariance estimators based on thresholding or penalization have also been widely studied for high-dimensional settings \cite{bickel2008covariance}. In many applications, particularly in graphical models, the inverse covariance matrix plays a fundamental role in capturing conditional dependence relationships between variables. Sparse precision matrix estimation methods such as the graphical lasso introduce an $\ell_1$ penalty to obtain sparse estimates of the precision matrix, which correspond to edges in an undirected graphical model \cite{friedman2008sparse}, see also \cite{ravikumar2011high,hsieh2011sparse}. Cholesky-based methods are also studied for covariance matrix and precision matrix estimation, see \cite{rothman2010new,khare2019scalable,kang2020cholesky,kang2020improved} and the references therein. Factor model based covariance estimators further exploit the assumption that the covariance structure is driven by a small number of latent factors, which substantially reduces the number of parameters to be estimated \cite{fan2013large}.

Another important line of research concerns robust covariance estimation, which aims to provide reliable estimates in the presence of outliers or heavy-tailed data distributions. Robust estimators include $M$-estimators of covariance matrices \cite{maronna1976robust}, the Minimum Covariance Determinant (MCD) estimator \cite{rousseeuw1985multivariate}, and other robust scatter estimators that maintain good statistical efficiency while limiting the influence of anomalous observations. These approaches are particularly relevant in financial and high-dimensional data analysis where classical estimators can be highly sensitive to deviations from Gaussian assumptions. For a review on this topic, see \cite{bai2011estimating,tong2014estimation,lam2020high}.

\section{Continuous-domain PCE (cPCE) and quantum learning of covariances}\label{Sec:cPCEQL}

In this section, we first introduce cPCE and then propose C-Estimator and E-Estimator for quantum learning of a classical covariance estimator.

\subsection{Continuous-domain PCE (cPCE)}

In this section we propose PCE for estimating variables that can take any values from the real line and we refer to it as continuous variables. Hence we extend the PCE technique to continuous-domain PCE (cPCE).

According to this we encode our continuous-domain variables as
\begin{equation}
    x_i:= c_i\langle\Pi_i\rangle,
\end{equation}
where $c_i$ is a normalization coefficient which we set to $c_i=1$ if search space in the optimization problem is known to be in the range $[-1, 1] \; \; \forall x_i$. $\Pi_i$ is defined according to the baseline defined in \ref{Sec:PCE}. 

Following this notation, the new loss function will consider continuous-domain values and can be adapted to any kind of problem. Next sections will introduce different loss functions in order to combine this approach with other use cases.


\subsection{cPCE-based covariance estimator formulation }

Now we describe our proposal for generation of a correlation estimator based on a quantum-classical machine learning model. The broad overview of the model is given by Figure \ref{fig:scheme}.

First note that there are ${n\choose 2}=n(n-1)/2$ variables to estimate in a covariance matrix corresponding to a random vector $X=(X_1,X_2,\hdots, X_n)$, which correspond to the strictly upper triangular part of the estimator $\widehat{\Sigma}.$ The matrix $\widehat{\Sigma}=[\widehat{\Sigma}_{ij}]$ being symmetric, $\widehat{\Sigma}_{ij}=\widehat{\Sigma}_{ji}$ for $i\neq j$ and $\widehat{\Sigma}_{ii}=V_i=Var(X_i),$ the variance of the $i$-th random variable $X_i.$ To employ cPCE for estimating the ${n\choose 2}$ parameters, we should consider a quantum register on $\eta_k=\ceil{{n\choose 2}^{1/k}}$ qubits, so that $(\eta_k)^k\approx {n\choose 2}.$ For instance, if $n=5$ then $\eta_2=\ceil{\sqrt{10}}=4$ and $\eta_3=\ceil{(10)^{1/3}}=3.$ The set of observables corresponding to $\eta_k$ is given by $\Pi^{(k)}=\left\{\Pi_1^{(k)}, \hdots, \Pi_{{\eta_k\choose 2}}^{(k)}\right\},$ where each  $\Pi^{(k)}_l$ is a permutation of a Pauli-string of the form $X^{\otimes k}\otimes I^{\otimes(\eta_k-k)}$, $Y^{\otimes k}\otimes I^{\otimes(\eta_k-k)}$, $Z^{\otimes k}\otimes I^{\otimes(\eta_k-k)}$ with $X,Y, Z$ are Pauli matrices of order $2.$ Now note that $\left|\Pi^{(k)}\right|=3{\eta_k\choose 2}$ for a value of $k.$ Indeed, note that  $$3{\eta_k \choose 2} = \dfrac{3\times \eta_k\times (\eta_k-1)}{2}=\dfrac{3}{2} {n\choose 2}^{1/k}\times \left({n\choose 2}^{1/k}-1\right).$$ Then for $k=2,$ $3{\eta_k\choose 2}=\dfrac{3}{2}{n\choose 2}-\sqrt{{n\choose 2}}=\frac{3}{2}{n\choose 2}-\mathcal{O}(n)$, and for $k=3$ we have $3{\eta_k\choose 2}=\frac{3}{2}\left({{n\choose 2}}^{2/3}-{{n\choose 2}}^{1/3}\right)\sim \mathcal{O}(n^{4/3})$ and ${n\choose 2}\sim \mathcal{O}(n^2).$ Consequently, for our purpose of using cPCE for covariance estimator, we set $k=2$ and we consider the case when $n$ is large such as $n\geq 10$ to guarantee that the number of ovservables for the number of qubits under consideration should be at least the desired number of parameters ${n\choose 2}.$   The comparison of the functions $f_k(\eta_k)=3{\eta_k\choose 2},$ $k=2,3$ and $g(n)={n\choose 2}$ is given in Figure \ref{fig:qubitcount}.  

\begin{figure}
    \centering
    \includegraphics[width=0.42\textwidth]{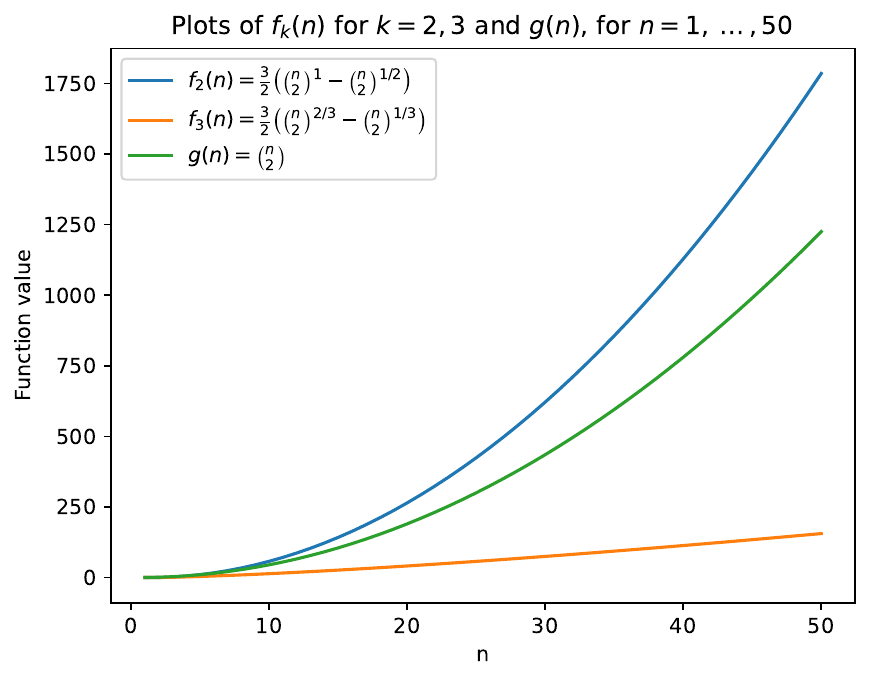}
    \caption{Growth of number of observables.}
    \label{fig:qubitcount}
\end{figure}


\subsubsection{C-Estimator}

In this section, we propose a covariance estimator through an estimator of its Cholesky factor, and we call it \textit{C-Estimator}. First, we discuss full C-Estimator, which means that a classical covariance matrix is known for all pairs of random variables $(X_i,X_j),$ $1\leq i\neq j\leq n.$ Next, we propose a C-Estimator when all covariance entries are not known. In what follows, we discuss C-Estimators using cPCE. 

First note that the diagonals of a covariance estimator are the variances corresponding to marginals, and the estimator must be a positive (semi)definite matrix. Thus we propose the estimator as $\widehat{\Sigma}=LL^\top,$ where $L$ is a lower-triangular matrix of order $n$ given by \begin{equation}\label{eqn:L}
L=\bmatrix{\lam_1& 0 & 0 & \hdots & 0 \\ c_{12}x_{12}& \lam_2& 0 & \ddots & 0\\
\vdots & &\ddots & &0\\
c_{1n}x_{1n}& c_{2n}x_{2n}& c_{3n}x_{3n}& \hdots & \lam_n},\end{equation} $\lam_i=\sqrt{\mbox{Var}(X_i)-\sum_{l=1}^{i-1} c_{li}^2x_{li}^2}$ for $i\geq 2,$ $\lam_1=\sqrt{\mbox{Var}(X_1)},$ $c_{ij}$s are regularization parameters, $x_{ij}=\bra{\psi(\boldsymbol{\theta})}\Pi_{ij}\ket{\psi(\boldsymbol{\theta})},$ $\Pi_{ij}\in \Pi^{(2)}.$ The parameters $c_{ij}$ should be chosen such that $\lam_i\geq 0.$ Then we have
$\widehat{\Sigma}=LL^\top=[\widehat{\Sigma}_{ij}],$ where $\widehat{\Sigma}_{ij}=\braket{e_iL|L^\top e_j}=\bra{e_i}L\cdot \bra{e_j}L,$ where $\ket{e_i}$ is the $i$-th column of the identity matrix of order $n$ and $u\cdot v$ denotes the scalar product of vectors. In particular, when $i=j,$ it can be easily verified that $\widehat{\Sigma}_{ii}=\mbox{Var}(X_i).$
Obviously, $\widehat{\Sigma}$ is positive semi-definite by construction, and it is positive definite when $\lam_i> 0$ for all $i.$ The regularization parameters $c_{ij}$ should be chosen such that \begin{equation}\label{eqn:cond1cij}\sum_{l=1}^{i-1} c_{li}^2x_{li}^2\leq \mbox{Var} (X_i).\end{equation} 


The optimal values of quantum circuit parameters $\boldsymbol{\theta}$ to obtain $\ket{\psi(\boldsymbol{\theta})}=U(\boldsymbol{\theta})\ket{{\bf 0}}$ by minimizing the loss function defined as follows. \begin{equation}\label{eqn:defloss1}
    \ell(\boldsymbol{\theta})= \sum_{\substack{i,j\\ j\leq i}}\left(\widehat{\Sigma}_{ij}-\widehat{\Sigma}^{ce}_{ij}\right)^2,
\end{equation} where $\widehat{\Sigma}^{ce}$ is the given classical estimator of the covariance matrix. 

\begin{algorithm}
  \caption{Quantum algorithm for a C-Estimator }
   {\bf Input:} $\widehat{\Sigma}^{ce},  \eta_2=\ceil{\sqrt{{n\choose 2}}},$ regularization parameters $c_{ij}$ \\
   {\bf Output:} $\widehat{\Sigma}$
  \begin{algorithmic}[1]
  \State Consider a $\eta_k$-qubit quantum register
    \State Consider a parametrized quantum circuit $U(\boldsymbol{\theta}),$ $\theta_j\in \R$ with the initial state $\ket{0}^{\otimes \eta_2}$ and set $\ket{\psi(\boldsymbol{\theta})}=U(\boldsymbol{\theta})\ket{0}^{\otimes \eta_2}$
    \State Choose ${n\choose 2}$ Pauli-strings $\Pi_{ij}\in \Pi^{(2)}$, $j< i$ and with a initial choice of the circuit parameters, estimate $x_{ij}=\braket{\psi(\boldsymbol{\theta})|\Pi_{ij}|\psi(\boldsymbol{\theta})}$ and compute first approximation $\widehat{\Sigma}=LL^\top$ (see Eqn. (\ref{eqn:L}))
    \State Update the circuit parameters using a classical optimizer to minimize $\ell(\boldsymbol{\theta})$ as defined in Eqn. (\ref{eqn:defloss1}) 
    \item Finally, determine $L_o$ corresponding to the optimal choice of the circuit parameters and return $\widehat{\Sigma}=L_oL_o^\top.$
  \end{algorithmic}
  \label{Alg:CEstimator}
\end{algorithm}

Recall that the inverse of the covariance matrix is known as the concentration or precision matrix and has a wide range of applications. Now we would like to explore to find an estimate for the inverse observing that the proposed estimator $\widehat{\Sigma}$ is obtained by estimating a Cholesky factor $L$ whose entries are estimated by the expected values of the observables in cPCE. This Cholesky factor can be further employed to obtain an estimate for the inverse of the covariance matrix efficiently using a classical algorithm. Indeed, note that if $\widehat{\Sigma}=LL^{\top}$ then $\widehat{\Sigma}^{-1}=(L^{-1})^\top L^{-1}.$ Since $L$ is lower-triangular, $L^{-1}$ can be obtained by solving two systems of equations $Ly=e_i$ and $L^\top x=y$ such that $x$ is the $i$-th column of $\widehat{\Sigma}^{-1}.$ The systems of equations can be solved using a substitution method, whose complexity is $\mathcal{O}(n^2).$ 
 
\subsubsection{Low-rank C-estimator}

In this section, we further propose a similar method for low-rank estimation of a given classical covariance estimator using cPCE. Indeed, note that rank of a psd matrix $A$ equals the rank of a Cholesky factor $L$ of the matrix, and the number of non-zero diagonal entries of $L$ is the rank of $L.$ Then we have the following proposal.

Let $\widehat{\Sigma}^{ce}$ be an estimated covariance matrix corresponding to a given model and we would like to find a rank-$r$ approximation of $\widehat{\Sigma}^{ce}.$ In that case we consider $L$ as described by equation (\ref{eqn:L}) and set $\lam_i=1$ for $i=1,\hdots,r$ and $\lam_i=0$ for $r<i\leq n.$  We denote it by $L_r$ and $L_rL_r^\top$ is a rank-$r$ approximation of $\widehat{\Sigma}^{ce}.$

The objective is to minimize the diagonal entries $\lam_i,$ $i=n-r+1,\hdots, n-1,n$ of the matrix $L.$ Thus minimizing loss function is formulated as \begin{eqnarray}
    && \min_{\boldsymbol{\theta}} \,\, \ell(\boldsymbol{\theta}) \nonumber\\
    &\mbox{s.t.}&  \mbox{Var}(X_i) = \sum_{l=1}^{i-1} c_{li}^2x_{li}^2, \,\, \mbox{for} \,\, i=n-r+1, \hdots, n-1,n \nonumber\\
    && \mbox{Var}(X_i) < \sum_{l=1}^{i-1} c_{li}^2x_{li}^2, \,\, \mbox{for} \,\, i=2, \hdots, n-r. \label{eqn:lossfs2}
\end{eqnarray}

\subsubsection{C-Estimator completion of partially specified covariance matrix} \label{Sec:partialCE}

Now we discuss the case when a covariance matrix in not known fully. Indeed, a partially specified covariance matrix is a matrix where only a subset of the variances and covariances (entries) are known. Then a challenging problem is requiring completion of the remaining entries to ensure the matrix is positive-definite and a meaningful covariance matrix to the data. Techniques to complete these matrices include the Expectation-Maximization (EM) algorithm and maximum-likelihood estimation (MLE) for structure-constrained data. In what follows, we discuss a cPCE based approach adapting the above method for completion of a partially specified covariance matrix.

Suppose that there are only $k$ entries of a classical estimator $\widehat{\Sigma}^{ce}$ are known, along with all the marginal variances. Suppose the known entries are indexed by $\widehat{\Sigma}^{ce}_{ij},$ $j<i$ where $$(i,j)\in S_l:=\{(i_1,j_1), (i_2,j_2), \hdots, (i_l,j_l)\}\subset [n]\times [n].$$ Then we formulate the loss function as \begin{equation}\label{eqn:lossf2}
    \ell_c(\boldsymbol{\theta}) = \sum_{(i,j)\in S_l}\left(\widehat{\Sigma}^{ij}-\widehat{\Sigma}^{ce}_{ij}\right)^2.
\end{equation}

\subsubsection{E-Estimator}

In this section, we propose a computationally efficient covariance estimator, which we call \textit{E-Estimator} using the cPCE technique. The key technique that we adapt in the model is estimating the covariance for the associated pair of random variables $(X_i, X_j),$ $1\leq i\neq j\leq n, $ by the scaled expectation values directly on a quantum register on $n$-qubits. Thus we designate: \begin{equation}
    \widehat{\Sigma}_{ij} = c_{ij}\bra{\psi(\boldsymbol{\theta})}\Pi_{ij}\ket{\psi(\boldsymbol{\theta})},
\end{equation} where $\Pi_{ij}$ is the Pauli $n$-string formed by Pauli $X$ matrix and identity matrix of order $2$ with the Pauli $X$ is placed at the $i$th and the $j$-th position of the string when $i\neq j$ and $\Pi_{ij}=I_{2^n}$ when $i=j$, $\ket{\psi(\boldsymbol{\theta})}=\mathsf{U}(\boldsymbol{\theta})\ket{0}^{\otimes n}$ for some parametrized unitary matrix $\mathsf{U}(\boldsymbol{\theta})\in\C^{2^n\times 2^n}$ defined by the set of parameters $\boldsymbol{\theta}.$  The parameters in $\boldsymbol{\theta}$ are learned and updated by minimizing the loss function $\mathcal{L}(\boldsymbol{\theta})$ given by equation (\ref{eqn:lossf}). 

Note that we compromise with the number of qubits  compared to finding a C-Estimator since $n>\eta_k.$ However, note that the observables are commuting, and hence the advantage is that they can be measured simultaneously since they share a common eigenbasis and allow efficient estimation of expectation values from the same measurement data. Further, note that the observables $\Pi_{ij}$ have the same set of eigenvalues $1$ and $-1,$ each having multiplicity $2^{n-1}.$ Consequently, we have a unitary matrix $E$ formed by the eigenvectors of $\Pi_{ij}$ such that $\Pi_{ij}E=E\Lambda_{ij},$ where $\Lambda_{ij}=P_{ij}\Lambda$ with $\Lambda$ is the diagonal matrix with entries $1$, $-1$ each is repeated $2^{n-1}$ times and $P_{ij}$ is a permutation matrix.

Indeed, recall that the orthonormal eigenpairs of the Pauli $X$ matrix are $(1,\ket{+})$ and $(-1,\ket{-1}),$ where $\ket{\pm}=(\ket{0}\pm \ket{1})/\sqrt{2}.$ Then it is easy to verify that the common eigenbasis of $\Pi_{ij}$ is given by $\{\ket{+}, \ket{-}\}^{\otimes n},$ however the eigenvalues need not be the same corresponding to a given eigenvector for different Pauli-strings. 

Thus the proposed E-Estimator is described in Equation~\ref{eq:Eestimator}. Consequently, $\widehat{\Sigma}$ is a positive semidefinite matrix if and only if $\widehat{\Lambda}$ is semidefinite, where $\Lambda_{ij}$ is a diagonal matrix with diagonal entries $\pm 1$ for all $i,j.$ Now we have the following theorem. 
\begin{figure*}[!t]
\centering
\begin{equation}
\label{eq:Eestimator}
\resizebox{0.94\textwidth}{!}{$
\widehat{\Sigma} =
\bmatrix{
\mathrm{Var}(X_1)I_{2^n} & \cdots & c_{1n}\braket{\psi(\boldsymbol{\theta})|\Pi_{1n}|\psi(\boldsymbol{\theta})} \\
c_{12}\braket{\psi(\boldsymbol{\theta})|\Pi_{12}|\psi(\boldsymbol{\theta})} & \cdots & c_{2n}\braket{\psi(\boldsymbol{\theta})|\Pi_{2n}|\psi(\boldsymbol{\theta})} \\
\vdots & \ddots & \vdots \\
c_{1n}\braket{\psi(\boldsymbol{\theta})|\Pi_{1n}|\psi(\boldsymbol{\theta})} & \cdots & \mathrm{Var}(X_n)I_{2^n}
}
=
\left(I_n\otimes\bra{\psi(\boldsymbol{\theta})}E\right)
\underbrace{
\bmatrix{
\mathrm{Var}(X_1)I_{2^n} & \cdots & c_{1n}\Lambda_{1n} \\
c_{12}\Lambda_{12} & \cdots & c_{2n}\Lambda_{2n} \\
\vdots & \ddots & \vdots \\
c_{1n}\Lambda_{1n} & \cdots & \mathrm{Var}(X_n)I_{2^n}
}
}_{:=\widehat{\Lambda}}
\left(I_n\otimes E^\dagger\ket{\psi(\boldsymbol{\theta})}\right)
$}.
\end{equation}
\end{figure*}

\begin{theorem}
The regularization parameters $c_{ij}$ can be chosen such that the E-Estimator $\widehat{\Sigma}$ is positive semidefine. In particular, choosing $c_{ij}$ such that \begin{equation}\label{eqn:cpsd}\sum_{\substack{l> k}} c_{kl} + \sum_{\substack{l<k}} c_{lk}\leq \mbox{Var} \, (X_k)\end{equation} for $1\leq k\leq n$ yields a positive semidefinite E-Estimator.    
\end{theorem}
\begin{proof}
Obviously, $\widehat{\Lambda}\in\R^{n2^{n}\times n2^{n}}$ is a sparse Hermitian matrix with sparsity $n$ i.e. exactly $n$ entries are non-zero for each row and column. Then $\widehat{\Lambda}$ is positive semidefinite if and only if all of its principal minors are non-negative. Thus sampling $c_{ij}$ ensuring the principal minors are non-negative proves the first part of the statement. The proof of the second part follows from the fact that $c_{ij}$ satisfying Eqn. (\ref{eqn:cpsd}) makes $\widehat{\Lambda}$ a diagonally dominant matrix. This completes the proof.       
\end{proof}

In the cPCE framework, training and updating the circuit parameters to estimate the E-Estimator $\widehat{\Sigma}$ is based minimizing the loss function  \begin{equation}\label{eqn:defloss2e}
    \mathcal{L}(\boldsymbol{\theta}) = \sum_{\substack{i<j\\i,j=0}}^{n-1} \left(c_{ij}\braket{\psi(\boldsymbol{\theta})|\Pi_{ij}|\psi(\boldsymbol{\theta})}  - \widehat{\Sigma}_{ij}^{ce}\right)^2.
\end{equation} As usual, $\widehat{\Sigma}^{ce}$ is a classical estimator obtained from real data of the random variables $X_1,\hdots, X_n.$ For instance, the linear shrinkage operator is given by $\widehat{\Sigma} := \alpha \Sigma_{SC} + (1-\alpha) I_n$ for some $\alpha\in (0, \, 1)$  and $\Sigma_{SC}$ is the sample covariance matrix. The Algorithm \ref{Alg:EEstimator} describes the QML based procedure for estimating an E-Estimator.

\begin{algorithm}
  \caption{Quantum algorithm for an E-Estimator }
   {\bf Input:} $\widehat{\Sigma}^{ce},$ regularization parameters $c_{ij}$ \\
   {\bf Output:} $\widehat{\Sigma}$
  \begin{algorithmic}[1]
  \State Consider an $n$-qubit quantum register
    \State Consider a parametrized quantum circuit $U(\boldsymbol{\theta}),$ $\theta_j\in \R$ with the initial state $\ket{0}^{\otimes n}$ and set $\ket{\psi(\boldsymbol{\theta})}=U(\boldsymbol{\theta})\ket{0}^{\otimes n}$
    \State Choose ${n\choose 2}$ Pauli-strings $\Pi_{ij}=I_2^{\otimes (i-1)}\otimes X\otimes I_2^{\otimes j-i-1}\otimes X\otimes I_2^{n-j}$, $i< j$ and with a initial choice of the circuit parameters, estimate $\braket{\Pi_{ij}}=\braket{\psi(\boldsymbol{\theta})|\Pi_{ij}|\psi(\boldsymbol{\theta})}$ 
    \State Update the circuit parameters using a classical optimizer to minimize $\mathcal{L}(\boldsymbol{\theta})$ as defined in Eqn. (\ref{eqn:defloss2e}) 
    \item Finally, determine $\boldsymbol{\theta}_o$ the optimal values of the circuit parameters and return $\widehat{\Sigma}.$
  \end{algorithmic}
  \label{Alg:EEstimator}
\end{algorithm}

\subsubsection{E-Estimator completion of partially specified covariance matrix}

Finding an E-Estimator for completing a partially specified covariance matrix follows a similar argument as in Section \ref{Sec:partialCE}.

\subsubsection{Analysis of the Algorithm \ref{Alg:EEstimator}}

Let us consider the classical covariance estimator $\Sigma^s=[\Sigma^s_{ij}].$ Considering the observable as permutations of $\Pi_{ij}^{(2)}=X^{\otimes 2}\otimes I^{\otimes (n-2)}$, the loss function $\mathcal{L}$ is given by \begin{equation}\label{eqn:lossf}
    \mathcal{L}({\boldsymbol{\theta}}) = \sum_{\substack{i<j\\ i,j=0}}^{n-1} \left(c_{ij}\tr[\rho(\boldsymbol{\theta})\Pi_{ij}] - \Sigma^s_{ij}\right)^2
\end{equation} which follows from equation (\ref{eqn:defloss2e}), where 
$\rho(\boldsymbol{\theta})=\mathsf{U}(\boldsymbol{\theta})\rho_0\mathsf{U}(\boldsymbol{\theta})^\dagger.$    
Then,
\begin{align}
\frac{\partial \mathcal{L}(\boldsymbol{\theta})}{\partial\theta}
&= 2\sum_{\substack{i<j\\ i,j=0}}^{n-1}
\left(c_{ij}\,\mathrm{tr}[c_{ij}\rho(\boldsymbol{\theta})\Pi_{ij}]
      - \Sigma^s_{ij}(\alpha)\right)
\frac{\partial}{\partial\theta}\mathrm{tr}[\rho(\boldsymbol{\theta})\Pi_{ij}]\nonumber
\\[6pt]
&= 2\sum_{\substack{i<j\\ i,j=0}}^{n-1}
\left(c_{ij}\,\mathrm{tr}[\rho(\boldsymbol{\theta})\Pi_{ij}]
      - \Sigma^s_{ij}(\alpha)\right)
\mathrm{tr}\!\left[\frac{\partial\rho(\boldsymbol{\theta})}{\partial\theta}\Pi_{ij}\right], \nonumber
\end{align}

Now from equation (\ref{eqn:uCircuit}), setting
\[
\mathsf{U}(\boldsymbol{\theta})
= \prod_{l=1}^L \prod_{k=1}^K e^{-iH_{lk}\theta_{lk}}
= \mathsf{U}_{<(l,k)} e^{-i\theta_{lk}H_{lk}} \mathsf{U}_{>(l,k)},
\]
we have
\begin{align}
\frac{\partial\mathsf{U}(\boldsymbol{\theta})}{\partial\theta_{lk}} \nonumber
&= -i\,\mathsf{U}_{<(l,k)} H_{lk}\nonumber
      e^{-i\theta_{lk}H_{lk}} \mathsf{U}_{>(l,k)} \\
&= -i\,\mathsf{U}_{<(l,k)} H_{lk} \mathsf{U}_{<(l,k)}^\dagger\nonumber
      \mathsf{U}(\boldsymbol{\theta}) = -i\,\widetilde{H}_{lk}(\boldsymbol{\theta})\nonumber
      \mathsf{U}(\boldsymbol{\theta}),
\end{align}
where
\[
\widetilde{H}_{lk}(\boldsymbol{\theta})
= \mathsf{U}_{<(l,k)} H_{lk} \mathsf{U}_{<(l,k)}^\dagger.
\]

Finally,
\begin{align}
\frac{\partial\rho(\boldsymbol{\theta})}{\partial\theta_{lk}}
&= \frac{\partial\mathsf{U}(\boldsymbol{\theta})}{\partial\theta_{lk}}\nonumber
   \rho_0 \mathsf{U}(\boldsymbol{\theta})^\dagger
   + \mathsf{U}(\boldsymbol{\theta})\rho_0
     \frac{\partial\mathsf{U}(\boldsymbol{\theta})^\dagger}{\partial\theta_{lk}} \\
&= -i\,[\,\widetilde{H}_{lk}(\boldsymbol{\theta}),\nonumber
         \rho(\boldsymbol{\theta})\,].
\end{align}

Finally,
\begin{align}
\frac{\partial \mathcal{L}(\boldsymbol{\theta})}{\partial\theta_{lk}}\nonumber
&= -2i
   \sum_{\substack{i<j \\ i,j=0}}^{N-1}
   \left(
     \mathrm{tr}[\rho(\boldsymbol{\theta})\Pi_{ij}]
     - \Sigma^s_{ij}(\alpha)
   \right)
   \mathrm{tr}\!\left(
       \rho(\boldsymbol{\theta})
       [\Pi_{ij}, \widetilde{H}_{lk}(\boldsymbol{\theta})]
   \right).
\end{align}

Now, from \cite{brandao2016local}, we know that local random quantum circuits with nearest-neighbor $2$-qubit gates and random $1$-qubit gates on $n$ qubits form an approximate unitary 
$t$-design for polynomial depth in $n$, with $t$ polynomially bounded. In our consideration of the HEA, which is formed by random $R_y$ gate on each qubit and $CZ$ gates between neighbors represents a local random circuit when repeated $L = poly(n)$ layers. Consequently, the first and second moments of any polynomial of degree $2$ in $\mathsf{U}$ and $\mathsf{U}^\dagger$ 
 corresponding to the HEA are (approximately) equal to those under the Haar measure. In \cite{mcclean2018barren}, the authors investigate the variance of the cost function defined by the expectation value over some Hermitian operator of the output of a random PQC especially for $2$-design and they show that the gradient of the variance scales exponentially vanishing in the number of qubits observing barren plateau phenomena i.e. Var$(\partial_\theta \mathcal{L})=\mathcal{O}(1/2^n)$. However, Cerezo et al. generalize this result in \cite{cerezo2021cost}  and investigate the barren plateau depending on the cepth of the circuit. Indeed, they show that the gradient of expectations of large observables (global cost functions) decays exponentially, whereas it decays at worst polynomially for local observables.  Now we recall the following results from \cite{cerezo2021cost} and \cite{mcclean2018barren} for any Haar random $\mathsf{U}\in \mathfrak{su}(d)$, a pure state $\rho$ and some operators $A,B$ that
 \begin{enumerate}
     \item First moment: $\mathbb{E}_U\left[\tr(\mathsf{U}\rho \mathsf{U}^\dagger A)\right]=\dfrac{\tr(A)}{d}$
     \item Second moment: $\mathbb{E}_U\left[\tr(\mathsf{U}\rho \mathsf{U}^\dagger A)\tr(\mathsf{U}\rho \mathsf{U}^\dagger B)\right]=\dfrac{\tr(AB)+\tr(A)\tr(B)}{d^2-1}$
     \item Scaling: $Var\left(\tr(\rho(\boldsymbol{\theta}))A\right)=\mathcal{O}\left(\dfrac{\|A\|^2_{\mathrm{HS}}}{d^2}\right)=\mathcal{O}\left(\dfrac{\|A\|^2_{\mathrm{HS}}}{4^n}\right),$ where $\|X\|_{\mathrm{HS}}=\sqrt{\tr(X^\dagger X)}.$  For bounded Hilbert-Schmidt norm $\|A\|^2_{\mathrm{HS}}=\mathcal{O}(poly(n))$  gives $Var\left(\tr(\rho(\boldsymbol{\theta}))A\right)=\mathcal{O}\left(\dfrac{poly(n)}{4^n}\right),$ for global cost functions. 
 \end{enumerate}

For brevity we consider the notation $r\in\{1,2,\hdots, n(n-1)/2\}$ corresponding to the pair $(i,j)$ in our consideration of the measurement operators $\Pi_{ij}^{(2)}$ as $\Pi_r$ and $\Sigma_{ij}^s(\alpha)=\Sigma_r^s(\alpha).$ Note that for our considered $n$-qubit HEA with $L$ layers on a ring, we have 
\begin{align}
\label{eqn:HEA}
\mathsf{U}_l(\boldsymbol{\theta}_l)
&=\nonumber
\left(\prod_{(i,i+1)\bmod n} CZ_{i,i+1}\right)
\left(\prod_{j=1}^n e^{-i\theta_{l,j} Y_j/2}\right),
\\[4pt]
\mathsf{U}(\boldsymbol{\theta})
&=
\mathsf{U}_L \cdots \mathsf{U}_1 .
\end{align}
Besides, $\Pi_r$ and each generator $Y_j/2$ has operator norm $\mathcal{O}(1).$ Moreover, note that a loss function is local if the observable involves only $\mathcal{O}(1)$ qubits such as $\mathcal{L}(\boldsymbol{\theta})=\tr(\rho(\boldsymbol{\theta})O)$ with $supp(O)$ is constant size. However, our loss function depends on all pair of qubits and hence it correlates across the entire system and hence it is a global cost function according to \cite{cerezo2021cost}. Thus we have the following theorem. 


\begin{theorem}\label{thm:Vlossf}
Consider the HEA as described in equation (\ref{eqn:HEA}) with $L=poly(n)$ layers which forms an approximate $2$-design. Then for the loss function $\mathcal{L}$ given by equation (\ref{eqn:lossf}), $$Var_\theta\left(\dfrac{\partial\mathcal{L}}{\partial\theta_{l,j}}\right)=\mathcal{O}\left(\dfrac{\|c\circ c\|_2^2+\|c\|_2^2\|\Sigma^s\|_F^2}{2^n}\right),$$ where $c=(c_1,\hdots.c_{n(n-1)/2})$ and $\Sigma^s=(\Sigma_1^s,\hdots,\Sigma^s_{n(n-1)/2}).$ In particular, when  the coefficients $c_r$ and $\Sigma^s_r$ are uniformly bonded i.e. $\max_r c_e$ and $\max_r\Sigma^s_r$ are $\mathcal{O}(1)$ then  $Var_\theta\left(\dfrac{\partial\mathcal{L}}{\partial\theta_{l,j}}\right)=\mathcal{O}\left(\dfrac{n^4}{2^n}\right).$ The notation $\circ$ denotes the Hadamard product of two vectors.

\end{theorem}
\pf Set $X_r(\boldsymbol{\theta})=\tr[\rho(\boldsymbol{\theta})\Pi_r]$ and $Y_{r,l,j}=\tr[\rho(\boldsymbol{\theta})[\Pi_r,\widetilde{H}_{l,j}(\boldsymbol{\theta})]].$ Then $$\dfrac{\partial\mathcal{L}}{\partial\theta_{l,j}}=-2i \sum_{r=1}^{n(n-1)/2} \left(c_rX_r(\boldsymbol{\theta} - \Sigma^s_r)c_r Y_{r,l,j}(\boldsymbol{\theta})\right).$$ Note that $\Pi_r$ and $[\Pi_r,\widetilde{H}_{l,j}]$, we have $\|\cdot\|_{\mathrm{op}}=\mathcal{O}(1)$, $\|\cdot\|^2=\mathcal{O}(2^n)$ with trace zero. This follows from the gact that $[\Pi_r,\widetilde{H}_{l,j}]$ is a sum of a constant number of Pauli strings acting on at most $4$ qubits. As discussed above, from a unitary $2$-design sample $U$ yields $$Var\left(\tr(U\rho U^\dagger A)\right)=\mathcal{O}\left(\frac{1}{2^n}\right).$$ Consequently, $A=\Pi_r$ implies $Var(X_r)=\mathcal{O}(2^{-n}),$ and $A=[\Pi_r,\widetilde{H}_{l,j}]$ implies $Var(Y_{r,l,j})=\mathcal{O}(2^{-n}).$ Similarly, $\mathbb{E}[X_r^2Y_{r,l,j}^2]$ scale as $\mathcal{O}(2^{-n})$ up to polynomial factors.

Now, using 
\[
\mathrm{Var}\!\left(\sum_r Z_r\right)
\le \left(\sum_r \sqrt{\mathrm{Var}(Z_r)}\right)^2
\]
and the Cauchy--Schwarz inequality, it follows that
\begin{align}
\mathrm{Var}\!\left(c_r^2 X_r Y_{r,l,j}\right) 
\le c_r^4\, \mathbb{E}[X_r^2 Y_{r,l,j}^2] = c_r^4\,\mathcal{O}\!\left(\frac{1}{2^n}\right), \nonumber
\end{align}
and
\begin{align}
\mathrm{Var}\!\left(c_r \Sigma_r^s Y_{r,l,j}\right) 
\le c_r^2\, \mathbb{E}[Y_{r,l,j}^2] = c_r^2\,(\Sigma_r^s)^2\,\mathcal{O}\!\left(\frac{1}{2^n}\right). \nonumber
\end{align}

Finally, summing over $r$ gives
\begin{align}
\mathrm{Var}\!\left(
   \frac{\partial\mathcal{L}}{\partial\theta_{l,j}}
\right) 
= \mathcal{O}\!\left(
   \frac{
      \sum_r c_r^4
      +
      \sum_r c_r^2 (\Sigma_r^s)^2
   }{2^n}
   \right) \nonumber
\\
= \mathcal{O}\!\left(
   \frac{
     \|c \circ c\|_2^2
     + \|c\|_2^2 \|\Sigma^s\|_F^2
   }{2^n}
   \right).
\end{align}
This completes the proof. It follows from Theorem \ref{thm:Vlossf} that the barren plateau can be mitigated when some of $c_r$ grow exponentially in $n.$ 

\section{Numerical simulations}
\label{sec_results}
In our experiments, we aim to evaluate the performance of the C-Estimator for estimating covariance matrices under various scenarios. For each experiment, we begin by generating a classical covariance matrix $\Sigma_\mathrm{ce} \in \mathbb{R}^{n \times n}$, ensuring it is positive semidefinite (PSD) by sampling a random matrix $A \in \mathbb{R}^{n \times n}$ with entries drawn from a standard normal distribution and computing $\Sigma_\mathrm{ce} = A A^\top$. In low-rank recovery experiments, we construct $\Sigma_\mathrm{ce}$ using a tall matrix $B \in \mathbb{R}^{n \times r}$ with $r < n$ to simulate a true rank-$r$ covariance via $\Sigma_\mathrm{ce} = B B^\top$. For covariance completion tests, we generate a random PSD matrix as above and randomly mask a fraction of its entries to simulate missing data. Each covariance matrix is then passed to the quantum covariance estimation procedure, which optimizes a parameterized quantum circuit to produce estimated matrix elements. For convergence experiments, we record the mean absolute error (MAE) between the estimated and original covariance matrices across iterations, allowing us to visualize both the average behavior and variability of the estimator over multiple independent runs.

\subsection{Convergence analysis}

\begin{figure}
    \centering
    \includegraphics[width=\linewidth]{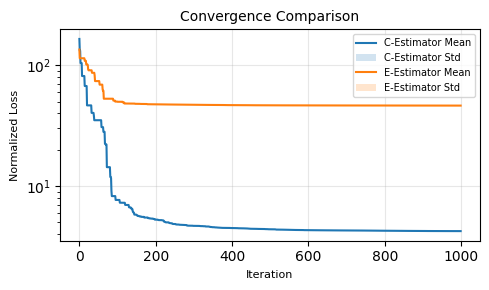}
    \caption{Convergence of the C-Estimator and E-Estimators during optimization, respectively. The solid line shows the mean of the best-so-far mean absolute error (MAE) across multiple runs. This illustrates the average convergence speed over iterations.}
    \label{fig:convergence_plot}
\end{figure}

To assess the convergence behavior of the estimators, we track the evolution of the mean absolute error (MAE) between the estimated covariance matrix and the original matrix over the optimization iterations. For each run, the MAE at iteration $t$ is recorded, and the minimum error achieved up to that iteration is computed, producing a "best-so-far" curve. When multiple independent runs are performed, the mean and standard deviation of the best-so-far MAE across runs are calculated. These quantities are visualized using a convergence plot in Figure~\ref{fig:convergence_plot}, where the mean curve is shown. This approach allows us to clearly observe a comparison between both estimators, highlighting its stability and robustness during the optimization procedure.

\subsection{Low-Rank Recovery}

Figure~\ref{fig:low_rank_recovery} evaluates the ability of the C-Estimator to recover covariance matrices that are approximately low-rank. A random covariance matrix $\Sigma \in \mathbb{R}^{n \times n}$ is generated. The C-Estimator is then applied assuming a range of target ranks $r \in [1, n]$. For each assumed rank, the estimator reconstructs the covariance from the available low-rank structure. The final loss, measured as the mean absolute error (MAE) between the reconstructed and true covariance, is recorded across multiple independent runs. This setup mimics scenarios where the underlying quantum correlations are effectively low-dimensional, and tests how sensitive the estimator is to rank assumptions.

Figure~\ref{fig:low_rank_recovery} presents the final loss as a function of the assumed rank. Each point represents the mean MAE over several independent runs, while the connecting line illustrates the trend. As expected, the error decreases as the assumed rank approaches the true rank $r_\text{true} = 2$, after which further increasing the rank yields minimal improvement. This behavior indicates that the estimator can effectively exploit low-rank structure when the rank is correctly identified.

\begin{figure}
    \centering
    \includegraphics[width=\linewidth]{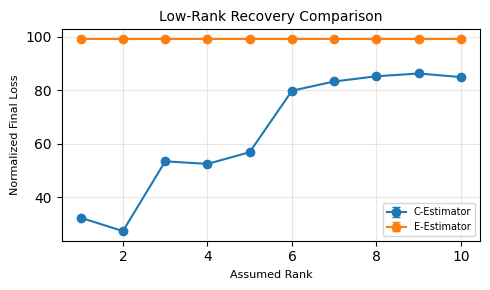}
    \caption{Low-rank covariance recovery. The plot shows the mean final MAE of the C-Estimator as a function of the assumed rank, with the true rank $r_\text{true}$ indicated by the dashed line. The error decreases when the assumed rank approaches the true rank, demonstrating the estimator's sensitivity to rank assumptions and its ability to exploit low-rank structure.}
    \label{fig:low_rank_recovery}
\end{figure}

\subsection{Covariance Completion Experiment}

This experiment tests the C-Estimator and E-Estimator's ability to reconstruct partially observed covariance matrices. A full covariance matrix $\Sigma \in \mathbb{R}^{n \times n}$ is first generated as $\Sigma = A A^\top$, where $A$ is a random matrix. A fraction of the matrix entries is then removed at random to simulate missing data, producing a masked matrix $\Sigma_\text{partial}$ and a corresponding mask $M$ indicating observed entries. The C-Estimator is applied to reconstruct the full covariance from the incomplete data. The final reconstruction error is measured using the mean absolute error (MAE) between the estimated covariance and the true covariance. This experiment evaluates the robustness of the estimator in the presence of missing information and its capability for covariance completion.

Figure~\ref{fig:covariance_completion} displays the final MAE as a function of the fraction of missing entries. Each point corresponds to the mean over multiple independent runs. As the fraction of missing data increases, the reconstruction error rises, reflecting the increasing difficulty of estimating the full covariance. The trend highlights the estimator's performance in handling incomplete datasets and shows how the accuracy degrades gracefully with more missing values.

\begin{figure}
    \centering
    \includegraphics[width=\linewidth]{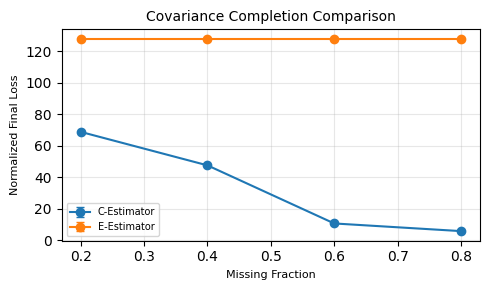}
    \caption{Covariance completion. The plot shows the mean final MAE of the C-Estimator (blue) and E-Estimator (orange) as a function of the fraction of missing entries, respectively. As expected, higher fractions of missing data lead to increased error, demonstrating the impact of incomplete observations on the covariance reconstruction.}
    \label{fig:covariance_completion}
\end{figure}

\subsection{Near-Singular Covariance Experiment} 
This experiment evaluates the stability and performance of the C-Estimator when applied to covariance matrices with very small eigenvalues. A near-singular covariance matrix $\Sigma \in \mathbb{R}^{n \times n}$ is constructed as $\Sigma = U \Lambda U^\top$, where $U$ is a random orthogonal matrix and $\Lambda = \mathrm{diag}(\lambda_1, \ldots, \lambda_n)$ is a diagonal matrix with eigenvalues spanning several orders of magnitude, from $\lambda_\text{max} = 1$ down to $\lambda_\text{min} = 10^{-4}$. This setup mimics scenarios where the covariance matrix is ill-conditioned or nearly rank-deficient, challenging the numerical stability of the estimator.

Figure~\ref{fig:near_singular} shows the convergence plot of the best loss value achieved so far as a function of optimization iterations, on a logarithmic scale. The shaded area represents the variability across multiple runs. The plot allows visualizing how quickly the estimator approaches a stable solution despite the ill-conditioning of the covariance matrix. A slower convergence or plateau at higher loss values indicates the numerical difficulty associated with near-singular structures.

\begin{figure}
    \centering
    \includegraphics[width=\linewidth]{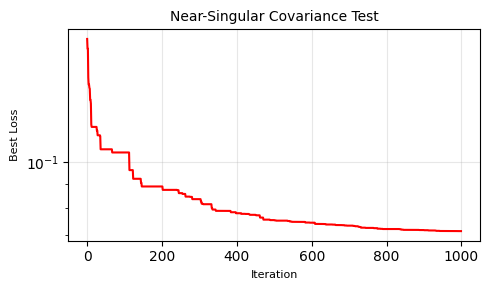}
    \caption{Near-singular covariance test. The figure shows the best-so-far MAE of the C-Estimator during optimization for a near-singular covariance matrix. The logarithmic scale highlights the convergence dynamics. Despite the ill-conditioning, the estimator achieves a stable reconstruction, illustrating its robustness to nearly rank-deficient matrices.}
    \label{fig:near_singular}
\end{figure}

\subsection{Gradient loss function E-Estimator}

We investigate the performance of the E-Estimator as a function of circuit depth by varying the number of layers as $\text{poly}(n)$ where $n$ is the number of qubits. Figure~\ref{fig_loss_gradient} shows for each number of qubits \(n \in \{4,6,8\}\) and number of layers the corresponding loss obtained. For each configuration, we generate a random target covariance matrix \(\Sigma_{\mathrm{CE}}\) and run the E-Estimator multiple times with random initial circuit parameters to account for stochastic variability. Scatter plot shows the mean loss across different runs. The resulting plot displays the loss as a function of the number of layers for each qubit count, allowing us to assess how circuit expressivity and depth impact the estimator's ability to recover the target covariance, and to observe potential saturation or optimization difficulties associated with deep circuits.

\begin{figure}
    \centering
    \includegraphics[width=\linewidth]{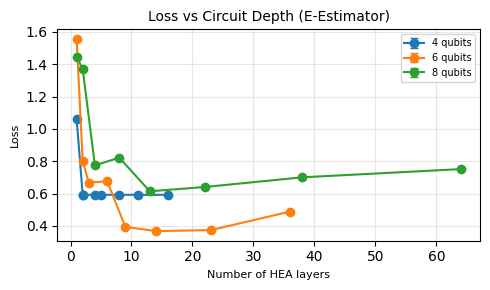}
    \caption{Converged loss as a function of poly number of layers for different number of qubits using cPCE approach with E-Estimator.}
    \label{fig_loss_gradient}
\end{figure}

\section{Conclusion}

Estimation of covariance matrices is a fundamental problem in statistics, machine learning, and financial data analysis, particularly in high-dimensional settings and in the presence of incomplete observations. In this work, we proposed a quantum machine learning framework for estimating classical covariance matrices using parameterized quantum circuits within the Pauli-Correlation-Encoding (PCE) paradigm. Two estimators, namely the C-Estimator and the E-Estimator, were introduced to learn covariance structures from observable expectation values of quantum states. The C-Estimator ensures positive (semi)definiteness through a Cholesky factorization and enables efficient computation of low-rank approximations and precision matrices, whereas the E-Estimator provides a computationally simpler approach for directly estimating covariance entries. We derived sufficient conditions on regularization parameters to guarantee positive (semi)definiteness of the learned covariance matrix and analyzed the trade-offs between the two estimators in terms of qubit requirements and learning complexity. Furthermore, we showed that appropriate choices of regularization parameters can mitigate the barren plateau phenomenon during training. Numerical simulations demonstrate that the proposed framework is robust for learning covariance matrices, handling low-rank structures, and solving covariance completion problems. These results suggest that quantum machine learning approaches may provide a promising direction for addressing high-dimensional statistical estimation problems in applications such as financial data analysis. 

Finally, we note that the proposed continuous-domain PCE framework offers a promising avenue for addressing a broad class of optimization problems, particularly in high-dimensional settings. By leveraging its ability to represent complex structures through quantum circuit encodings, the framework can be utilized to approximate solutions where traditional methods face scalability challenges. In particular, it may enable the encoding of high-dimensional algebraic manifolds via expectation values of suitably chosen observables, thereby providing a representation of the solution space. This perspective may open the door to new algorithmic strategies that combine probabilistic modeling with geometric insights, potentially leading to more efficient and scalable quantum optimization techniques.\\\\

{\bf Acknowledgment.} The authors acknowledge helpful discussions on this topic with colleagues in Fujitsu Research of Europe, Fujitsu Research of America, and Fujitsu Research of Japan. Additionally, all the numerical simulation were run using Fujitsu QARP software.

\bibliographystyle{plain}
\bibliography{reff}
\end{document}